\documentclass[12pt]{article}
\usepackage[left=35mm,right=35mm,top=35mm,bottom=35mm]{geometry}
\usepackage{amsmath}
\usepackage{amsthm}
\usepackage{amssymb}
\usepackage{amsfonts}
\usepackage{mathrsfs}

\DeclareMathOperator*{\slim}{s-lim}

\DeclareMathOperator{\supp}{supp}

\makeatletter

\@addtoreset{equation}{section}
\makeatother

\newtheorem{Prop}{Proposition}[section]
\newtheorem{Ass}[Prop]{Assumption}
\newtheorem{Thm}[Prop]{Theorem}
\newtheorem{Rem}[Prop]{Remark}
\newtheorem{Lem}[Prop]{Lemma}

\title{Quantum inverse scattering for time-decaying harmonic oscillators}
\author{Atsuhide ISHIDA\\
\\
Department of Liberal Arts, Faculty of Engineering,
 \\Tokyo University of Science\\
\normalsize 6-3-1 Niijuku, Katsushika-ku,Tokyo 125-8585, Japan\\
\normalsize E-mail: aishida@rs.tus.ac.jp\\
\normalsize Fax: +81-3-5876-1616
}
\date{}

\begin{document}
\begin{flushleft}
{\Large \bf Quantum inverse scattering for time-decaying harmonic oscillators}
\end{flushleft}

\begin{flushleft}
{\large Atsuhide ISHIDA}\\
{Katsushika Division, Institute of Arts and Sciences, Tokyo University of Science, 6-3-1 Niijuku, Katsushika-ku, Tokyo 125-8585, Japan\\ 
Alfr\'ed R\'enyi Institute of Mathematics, Re\'altanoda utca 13-15, Budapest 1053, Hungary\\
Email: aishida@rs.tus.ac.jp
}
\end{flushleft}

\begin{abstract}
Different from the usual harmonic oscillator, the time-decaying harmonic oscillator accelerates particles and generates scattering states. We study one of the multidimensional inverse scatterings in this two-body quantum system perturbed by short-range potential functions that have a bounded part and a locally singular part. Applying the Enss--Weder time-dependent method, we prove that the scattering operator determines the potential functions uniquely.
\end{abstract}


\section{Introduction\label{introduction}}
Let $n\geqslant2$, $x=(x_1,\ldots,x_n)\in\mathbb{R}^n$ and ${\rm i}p=(\partial_{x_1},\ldots,\partial_{x_n})$ with ${\rm i}=\sqrt{-1}$. In this paper, we consider the quantum system governed by the following time-dependent free Hamiltonian
\begin{equation}
H_0(t)=p^2/2+k(t)x^2/2
\end{equation}
acting on $L^2(\mathbb{R}^n)$, where the time-decay coefficient of the harmonic term is
\begin{equation}
k(t)=
\begin{cases}
\ \omega^2 & \quad \mbox{if}\quad|t|<r_0,\\
\ \sigma/t^2 & \quad \mbox{if}\quad|t|\geqslant r_0.\label{coeff}
\end{cases}
\end{equation}
for $0<\sigma<1/4$, $\omega>0$ and $r_0>0$. For simplicity, we write
\begin{equation}
0<\lambda=(1-\sqrt{1-4\sigma})/2<1/2.
\end{equation}
We now state the assumptions imposed on the potential functions as multiplication operators that perturb $H_0(t)$.
\begin{Ass}\label{ass}
The potential function $V$ is decomposed into a bounded part and a singular part,
\begin{equation}
V(x)=V^{\rm bdd}(x)+V^{\rm sing}(x).
\end{equation}
$V^{\rm bdd}\in L^\infty(\mathbb{R}^n)$ satisfies
\begin{equation}
|V^{\rm bdd}(x)|\lesssim\langle x\rangle^{-\rho}\label{bdd}
\end{equation}
almost everywhere, with $x\in\mathbb{R}^n$, $\rho>1/(1-\lambda)$, $\langle\cdot\rangle=\sqrt{1+|\cdot|^2}$ and $A\lesssim B$ means that there exists a constant $C>0$ such that $A\leqslant CB$. $V^{\rm sing}\in L^q(\mathbb{R}^n)$ is compactly supported, where $q$ satisfies
\begin{equation}
\infty>q
\begin{cases}
\ =2\quad & \mbox{\rm if}\quad n\leqslant3,\\
\ >n/2\quad & \mbox{\rm if}\quad n\geqslant4.\label{sing}
\end{cases}
\end{equation}
\end{Ass}
The part $V^{\rm sing}$ is well known to be $p^2$-bounded infinitesimally (\cite[Example 4.53]{LoHiBe}). We define the full Hamiltonian such that
\begin{equation}
H(t)=H_0(t)+V(x).
\end{equation}
The Newton equation of classical mechanics
\begin{equation}
({\rm d}^2/{\rm d}t^2)x(t)=-k(t)x(t)\label{newton}
\end{equation}
has general solution $x(t)=c_1t^{1-\lambda}+c_2t^\lambda$ for $t\geqslant r_0$ and the classical trajectory of the free particle behaves like $x(t)=O(t^{1-\lambda})$ as $t\rightarrow\infty$. From this classical motion of the particle, Ishida--Kawamoto \cite[Theorems 1 and 2]{IsKa1} proved that the threshold between short- and long-range is $-1/(1-\lambda)$.\par
The spectrum of $H_0(t)$ is absolutely continuous for $|t|\geqslant r_0$ because $H_0(t)$ is unitary equivalent with $p^2/(2|t|^{2\lambda})$ that is the free part of \eqref{reduced_h}. By virtue of Yajima \cite[Theorem 6 and Remark (a)]{Ya}, the existence of the propagators uniquely generated by $H_0(t)$ and $H(t)$ is guaranteed under Assumption \ref{ass}. We denote these propagators by $U_0(t,s)$ and $U(t,s)$, respectively. The wave operators
\begin{equation}
W^\pm=\slim_{t\rightarrow\pm\infty}U(t,0)^*U_0(t,0)\label{waveop}
\end{equation}
then exist by \cite[Theorem 1]{IsKa1} and the scattering operator is defined such that
\begin{equation}
S(V)=(W^+)^*W^-.
\end{equation}

\begin{Rem}
We can also treat the following combination of $V^{\rm bdd}$- and $V^{\rm sing}$-type potentials $V=V_{\rm sing}$ that satisfy $V_{\rm sing}\in L^q(\mathbb{R}^n)$ with \eqref{sing}, and $\langle x\rangle^\rho V_{\rm sing}(x)\langle p\rangle^{-2}$ is the bounded operator on $L^2(\mathbb{R}^n)$ for $\rho>1/(1-\lambda)$. To prove Theorem \ref{main_thm} for this $V=V_{\rm sing}$, it suffices to modify slightly the proof of Lemma \ref{lem3} (specifically, \eqref{lem3_3}, \eqref{lem3_5}, and \eqref{lem3_12}).
\end{Rem}

\begin{Rem}
\cite[Theorem 1]{IsKa1} proves the existence of \eqref{waveop} only for $V=V^{\rm bdd}$. We can immediately prove the existence of \eqref{waveop} for $V=V^{\rm bdd}+V^{\rm sing}$ and $V=V_{\rm sing}$ using propagation estimates \cite[Proposition 2]{IsKa1} and the standard Cook-Kuroda method \cite[Theorem XI.4]{ReSi2}.
\end{Rem}

\begin{Rem}
If the constants $\omega$, $r_0$ and $\lambda$ satisfy the following relation
\begin{equation}
\omega r_0\tan\omega r_0=-\lambda,\label{tan_relation}
\end{equation}
the ordinary differential equation \eqref{newton} has two fundamental solutions for all $t\in\mathbb{R}$, and by \cite[Theorem 1.2]{KaYo} (see also \cite[Lemma 3.1]{Ko}), the $L^{p_1}L^{p_2}$-type estimates
\begin{equation}
\|U_0(t,0)\phi\|_{L^{p_1}}\lesssim|t|^{-n(1-\lambda)(1/2-1/p_1)}\|\phi\|_{L^{p_2}}
\end{equation}
holds under \eqref{tan_relation} for $\phi\in\mathscr{S}(\mathbb{R}^n)$, which is the rapid decreasing function space, where $p_1\geqslant2$ and $p_2$ is the H\"older conjugate of $p_1$. Note that we do not have to assume \eqref{tan_relation} and any other relations among $\omega$, $r_0$ and $\lambda$ in this paper.
\end{Rem}

Applying the Enss--Weder time-dependent method \cite{EnWe}, we prove the following theorem that claims that the scattering operator determines the potential functions uniquely.
\begin{Thm}\label{main_thm}
Let $V_1$ and $V_2$ satisfy Assumption \ref{ass}. If $S(V_1)=S(V_2)$, then $V_1=V_2$ holds.
\end{Thm}

Since the Enss--Weder time-dependent method was devised, many authors have applied it to establish the uniqueness of the potential functions for various quantum models. \cite{AdKaKaTo}, \cite{AdFuIs}, \cite{AdMa}, \cite{Is2}, \cite{Ni1}, \cite{Ni2}, \cite{VaWe}, and \cite{We} investigated the models with external electric fields, whereas \cite{Is1}, \cite{Is4} and \cite{Ni3} studied repulsive Hamiltonians, and \cite{Is3} and \cite{Ju} studied fractional and relativistic Laplacians. \cite{Wa1}, \cite{Wa2}, and \cite{Wa3} applied the method to the non-linear Schr\"odinger equations and the Hartree-Fock equations.\par
The time-decaying harmonic oscillator has been an interesting topic for research in mathematical aspect despite the fact that this model does not come from any concrete physical phenomena. For the usual harmonic oscillator, there are no scattering states and all of its spectrum is covered by the infinite discrete pure points. However, if the term $x^2$ has a time-decay coefficient of a specified order, the situation changes completely. This time-decay coefficient accelerates the particles and generates the scattering states. From this perspective, \cite{IsKa1} and \cite{IsKa2} discussed whether the wave operators exist. \cite{IsKa1} proved that $V(x)=O(|x|^{-\rho_{\rm L}})$ as $|x|\rightarrow\infty$ with $0<\rho_{\rm L}\leqslant1/(1-\lambda)$ has to belong to the long-range class and proposed Dollard-type modified wave operators. If $\sigma=1/4$ in \eqref{coeff}, the circumstances of scattering change considerably. \cite{IsKa2} found that the classical trajectory has order $x(t)=\sqrt{t}\log t$ as $t\rightarrow\infty$ and clarified the threshold between short- and long-range. In contrast, \cite{KaYo} and \cite{Ka3} constructed the Strichartz estimates, and recent studies of non-linear analysis \cite{Ka1}, \cite{Ka2}, \cite{Ka3}, \cite{KaMi}, \cite{KaMu} and \cite{KaSa} have shown progress.\par
By the definition of $k(t)$, $H_0(t)\equiv H_0=p^2/2+\omega^2x^2/2$ is a time-independent harmonic oscillator for $|t|<r_0$. Let us here state the well-known Mehler formula for $H_0$ (\cite[Section 2.2]{BoCaHaMi} or \cite[Theorem 5.25]{LoHiBe}); specifically, the time evolution for $0<|t|<r_0$ and $\omega t\notin\pi\mathbb{Z}$ is represented as
\begin{equation}
e^{-{\rm i}tH_0}=\mathscr{M}(\tan\omega t/\omega)\mathscr{D}(\sin\omega t/\omega)\mathscr{F}\mathscr{M}(\tan\omega t/\omega)\label{mehler1}
\end{equation}
where $\mathscr{M}$ denotes multiplication and $\mathscr{D}$ denotes dilation,
\begin{gather}
\mathscr{M}(t)\phi(x)=e^{{\rm i}x^2/(2t)}\phi(x),\\
\mathscr{D}(t)\phi(x)=({\rm i}t)^{-n/2}\phi(x/t),
\end{gather}
and $\mathscr{F}$ denotes the Fourier transform over $L^2(\mathbb{R}^n)$ defined by the dense extension of
\begin{equation}
\mathscr{F}\phi(\xi)=\int_{\mathbb{R}^n}e^{-{\rm i}x\cdot\xi}\phi(x)dx/(2\pi)^{n/2}
\end{equation}
for $\phi\in\mathscr{S}(\mathbb{R}^n)$. A straightforward calculation yields
\begin{gather}
\mathscr{D}(\sin\omega t/\omega)={\rm i}^{n/2}\mathscr{D}(\cos\omega t)\mathscr{D}(\tan\omega t/\omega),\\
\mathscr{M}(\tan\omega t/\omega)\mathscr{D}(\cos\omega t)\mathscr{M}(-\tan\omega t/\omega)=\mathscr{M}(-\cot\omega t/\omega)\mathscr{D}(\cos\omega t)
\end{gather}
and
\begin{equation}
e^{-{\rm i}tH_0}={\rm i}^{n/2}\mathscr{M}(-\cot\omega t/\omega)\mathscr{D}(\cos\omega t)e^{-{\rm i}\tan\omega tp^2/(2\omega)}\label{mehler2}
\end{equation}
for $\omega t\notin(\pi/2)\mathbb{Z}$ because
\begin{equation}
e^{-{\rm i}tp^2/2}=\mathscr{M}(t)\mathscr{D}(t)\mathscr{F}\mathscr{M}(t).
\end{equation}
The formula \eqref{mehler2} was derived originally in Ishida \cite{Is1} for the repulsive Hamiltonian. On the other hand, $U_0(t,s)$ and $U(t,s)$ also have the convenient factorizations for $t,s\geqslant r_0$ or $t,s\leqslant-r_0$, that were proved by \cite[Proposition 1]{IsKa1}. We define
\begin{equation}
\tilde{U}_0(t)=e^{{\rm i}\lambda x^2/(2t)}e^{-{\rm i}\lambda\log tA}e^{-{\rm i}t^{1-2\lambda}p^2/(2(1-2\lambda))}\label{factorization}
\end{equation}
if $t\geqslant r_0$ and
\begin{equation}
\tilde{U}_0(t)=e^{{\rm i}\lambda x^2/(2t)}e^{-{\rm i}\lambda\log(-t)A}e^{{\rm i}(-t)^{1-2\lambda}p^2/(2(1-2\lambda))}
\end{equation}
if $t\leqslant-r_0$, where $A=(p\cdot x+x\cdot p)/2$. Then
\begin{equation}
U_0(t,s)=\tilde{U}_0(t)\tilde{U}_0(s)^*
\end{equation}
and
\begin{equation}
U(t,s)=e^{{\rm i}\lambda x^2/(2t)}e^{-{\rm i}\lambda\log|t|A}\hat{U}(t,s)e^{{\rm i}\lambda\log|s|A}e^{{-\rm i}\lambda x^2/(2s)}
\end{equation}
hold for $t,s\geqslant r_0$ or $t,s\leqslant-r_0$, where $\hat{U}(t,s)$ is the propagator generated by
\begin{equation}
\hat{H}(t)=p^2/(2|t|^{2\lambda})+V(|t|^\lambda x).\label{reduced_h}
\end{equation}
We additionally define
\begin{equation}
\tilde{U}_0(t)=e^{-{\rm i}tH_0}
\end{equation}
if $|t|<r_0$. The following strong limits
\begin{equation}
\tilde{W}^{\pm}=\slim_{t\rightarrow\pm\infty}U(t,0)^*\tilde{U}_0(t)
\end{equation}
exist because \eqref{waveop} exist and we define
\begin{equation}
\tilde{S}(V)=(\tilde{W}^+)^*\tilde{W}^-.
\end{equation}
Noting that $W^\pm=\tilde{W}^\pm\tilde{U}_0(s_\pm)^*U_0(s_\pm,0)$ for $s_+\geqslant r_0$ and $s_-\leqslant-r_0$, we easily find that $S$ and $\tilde{S}$ the relation
\begin{equation}
S(V)=U_0(s_+,0)^*\tilde{U}_0(s_+)\tilde{S}(V)\tilde{U}_0(s_-)^*U_0(s_-,0),\label{tilde_S}
\end{equation}
and that $S(V_1)=S(V_2)$ is equivalent to $\tilde{S}(V_1)=\tilde{S}(V_2)$. To analyze the time-evolution by $U_0(t,0)$ for all $t\in\mathbb{R}$ directly is difficult but can be overcome by pursuing the evolution of $e^{-{\rm i}\tan\omega tp^2/(2\omega)}$ if $|t|<r_0$ and $e^{\mp{\rm i}|t|^{1-2\lambda}p^2/(2(1-2\lambda))}$ if $|t|\geqslant r_0$. While the particles escape with order $x(t)=O(|t|^{1-\lambda})$ through the potential effects when $|t|\geqslant r_0$, the particles cannot scatter far away from the time-independent harmonic oscillator when $|t|<r_0$. We consequently reconstruct the potential functions from the time-independent harmonic oscillator (see the proofs of Theorems \ref{the1} and \ref{the2}).\par

Throughout this paper, we use the following notation; $\|\cdot\|$ denotes the $L^2$-norm and operator norm on $L^2(\mathbb{R}^n)$, $(\cdot,\cdot)$ the scalar product of $L^2(\mathbb{R}^n)$, and $F(\cdots)$ the characteristic function of the set $\{\cdots\}$. 

\section{Bounded case}
We first consider the instances for which $V^{\rm sing}\equiv0$, that is, $V=V^{\rm bdd}$ and prove the following reconstruction formula. 

\begin{Thm}\label{the1}
Let $\Phi_0\in\mathscr{S}(\mathbb{R}^n)$ such that $\mathscr{F}\Phi_0\in C_0^\infty(\mathbb{R}^n)$. For $v\in\mathbb{R}^n$, its normalization is $\hat{v}=v/|v|$. Let $\Phi_v=e^{{\rm i}v\cdot x}\Phi_0$ and $\Psi_v$ have the same properties. Then
\begin{equation}
\lim_{|v|\rightarrow\infty}|v|({\rm i}(\tilde{S}(V^{\rm bdd})-1)\Phi_v,\Psi_v)=\int_{-\infty}^\infty(V^{\rm bdd}(x+\hat{v}t)\Phi_0,\Psi_0){\rm d}t
\end{equation}
holds.
\end{Thm}

We now prepare to prove Theorem \ref{the1}. The following propagation estimates for the free evolution $e^{-{\rm i}tp^2/2}$ \cite[Proposition 2.10]{En} is very useful in some of our estimates.

\begin{Prop}\label{enss}
Let $M$ and $M'$ be measurable subsets of $\mathbb{R}^n$ and $f\in C_0^\infty(\mathbb{R}^n)$ have $\supp f\subset\{\xi\in\mathbb{R}^n\bigm||\xi|\leqslant\eta\}$ for some $\eta>0$. Then
\begin{equation}
\|F(x\in M')e^{-{\rm i}tp^2/2}f(p)F(x\in M)\|\lesssim_{N,f}(1+|t|+r)^{-N}
\end{equation}
for $t\in\mathbb{R}$ and $N\in\mathbb{N}$, where $r={\rm dist}(M',M)-\eta|t|\geqslant0$ and $\lesssim_{N,f}$ means that the constant depends on $N$ and $f$.
\end{Prop}

The following Lemma is the key propagation estimate in this section.

\begin{Lem}\label{lem1}
Let $\Phi_v$ be as in Theorem \ref{the1}. Then
\begin{equation}
\int_{-\infty}^{\infty}\|V^{\rm bdd}(x)\tilde{U}_0(t)\Phi_v\|{\rm d}t=O(|v|^{-1})\label{lem1_0}
\end{equation}
holds as $|v|\rightarrow\infty$.
\end{Lem}

\begin{proof}
We can take $f\in C_0^\infty(\mathbb{R}^n)$ such that $\Phi_0=f(p)\Phi_0$ and $\supp f\subset\{\xi\in\mathbb{R}^n\bigm||\xi|\leqslant\eta\}$ with some $\eta>0$. We separate the integral such that
\begin{equation}
\int_{-\infty}^{\infty}=\int_{|t|<r_0}+\int_{|t|\geqslant r_0}
\end{equation}
and consider $|t|<r_0$ first. By \eqref{mehler2} and the relation
\begin{equation}
e^{-{\rm i}v\cdot x}e^{-{\rm i}\tan\omega tp^2/(2\omega)}e^{{\rm i}v\cdot x}=e^{-{\rm i}\tan\omega t|v|^2/(2\omega)}e^{-{\rm i}\tan\omega tp\cdot v/\omega}e^{-{\rm i}\tan\omega tp^2/(2\omega)},\label{lem1_1}
\end{equation}
we have
\begin{gather}
\|V^{\rm bdd}(x)e^{-{\rm i}tH_0}\Phi_v\|=\|V^{\rm bdd}(\cos\omega tx)e^{-{\rm i}\tan\omega tp^2/(2\omega)}\Phi_v\|\nonumber\\
=\|V^{\rm bdd}(\cos\omega tx+\sin\omega tv/\omega)e^{-{\rm i}\tan\omega tp^2/(2\omega)}\Phi_0\|\leqslant I_1+I_2+I_3,\label{lem1_2}
\end{gather}
where we put
\begin{align}
I_1&=\|V^{\rm bdd}(x)\|\|F(|x|\geqslant|\tan\omega t||v|/(2\omega))e^{-{\rm i}\tan\omega tp^2/(2\omega)}f(p)\nonumber\\
&\qquad\times F(|x|\leqslant|\tan\omega t||v|/(4\omega))\|\|\Phi_0\|,\nonumber\\
I_2&=\|V^{\rm bdd}(x)\|\|F(|x|\geqslant|\tan\omega t||v|/(2\omega))e^{-{\rm i}\tan\omega tp^2/(2\omega)}f(p)\nonumber\\
&\qquad\times F(|x|>|\tan\omega t||v|/(4\omega))\langle x\rangle^{-2}\|\|\langle x\rangle^2\Phi_0\|,\qquad\nonumber\\
I_3&=\|V^{\rm bdd}(\cos\omega tx+\sin\omega tv/\omega)F(|x|<|\tan\omega t||v|/(2\omega))\|\|\Phi_0\|\label{lem1_3}
\end{align}
as in the proof of \cite[Proposition 2.2]{Is1} (see also \cite{AdKaKaTo}, \cite{AdFuIs}, \cite{AdMa}, \cite{EnWe}, \cite{Is2}, \cite{Is3}, \cite{VaWe}, and \cite{We}). Because of the periodicity of $\tan\omega t$, we can assume that
\begin{equation}
\pi/(2\omega)\leqslant r_0<\pi/\omega\label{r_0}
\end{equation}
without loss of generality. Moreover, if $r_0<\pi/(2\omega)$, we can demonstrate our proofs much more simply. We state this details in Remark \ref{rem2}. Using Proposition \ref{enss} for $I_1$, we have
\begin{equation}
\int_{|t|<r_0}(I_1+I_2){\rm d}t\lesssim\int_0^{\pi/(2\omega)}+\int_{\pi/(2\omega)}^{r_0}\langle\tan\omega tv\rangle^{-2}{\rm d}t.
\end{equation}
When $0\leqslant t<\pi/(2\omega)$, $\tan\omega t\geqslant\omega t$ and
\begin{equation}
\int_0^{\pi/(2\omega)}\langle\tan\omega tv\rangle^{-2}{\rm d}t\leqslant\int_0^{\pi/(2\omega)}\langle\omega tv\rangle^{-2}{\rm d}t=|v|^{-1}\int_0^{\pi|v|/(2\omega)}\langle\omega\tau\rangle^{-2}{\rm d}\tau=O(|v|^{-1})\label{lem1_4}
\end{equation}
hold by changing $\tau=t|v|$. When $\pi/(2\omega)\leqslant t<r_0$, $|\tan\omega t|>\pi-\omega t$ and
\begin{gather}
\int_{\pi/(2\omega)}^{r_0}\langle\tan\omega tv\rangle^{-2}{\rm d}t\leqslant\int_{\pi/(2\omega)}^{r_0}\langle(\pi-\omega t)v\rangle^{-2}{\rm d}t\nonumber\\
=|v|^{-1}\int_{(\pi/\omega-r_0)|v|}^{\pi|v|/(2\omega)}\langle\omega\tau\rangle^{-2}{\rm d}\tau=O(|v|^{-2})\label{lem1_5}
\end{gather}
hold by changing $\tau=(\pi/\omega-t)|v|$. As for $I_3$, when $|x|<|\tan\omega t||v|/(2\omega)$,
\begin{equation}
|\cos\omega tx+\sin\omega tv/\omega|>|\sin\omega t||v|/(2\omega)
\end{equation}
and
\begin{equation}
I_3\leqslant\|V^{\rm bdd}(x)F(|x|>|\sin\omega t||v|/(2\omega))\|\|\Phi_0\|
\end{equation}
hold. Assuming \eqref{bdd}, we have
\begin{equation}
\int_{|t|<r_0}I_3{\rm d}t\lesssim\int_0^{\pi/(2\omega)}+\int_{\pi/(2\omega)}^{r_0}\langle\sin\omega tv\rangle^{-\rho}{\rm d}t=O(|v|^{-1})+O(|v|^{-\rho}),\label{lem1_6}
\end{equation}
noting that $\rho>1/(1-\lambda)>1$ because $\sin\omega t\geqslant\omega t/2$ when $0\leqslant t<\pi/(2\omega)$, and $\sin\omega t>(\pi-\omega t)/2$ when $\pi/(2\omega)\leqslant t<r_0$. We next consider the integral over $|t|\geqslant r_0$, in particular, we consider $t\geqslant r_0$. Integral over $t\leqslant-r_0$ can be estimated in the same way with $t\geqslant r_0$. By \eqref{factorization} and relation
\begin{gather}
e^{-{\rm i}v\cdot x}e^{-{\rm i}t^{1-2\lambda}p^2/(2(1-2\lambda))}e^{{\rm i}v\cdot x}\nonumber\\
=e^{-{\rm i}t^{1-2\lambda}|v|^2/(2(1-2\lambda))}e^{-{\rm i}t^{1-2\lambda}p\cdot v/(1-2\lambda)}e^{-{\rm i}t^{1-2\lambda}p^2/(2(1-2\lambda))},
\end{gather}
we have
\begin{gather}
\|V^{\rm bdd}(x)\tilde{U}_0(t)\Phi_v\|=\|V^{\rm bdd}(t^\lambda x)e^{-{\rm i}t^{1-2\lambda}p^2/(2(1-2\lambda))}\Phi_v\|\nonumber\\
=\|V^{\rm bdd}(t^\lambda x+t^{1-\lambda}v/(1-2\lambda))e^{-{\rm i}t^{1-2\lambda}p^2/(2(1-2\lambda))}\Phi_0\|\lesssim I_4+I_5+I_6,\label{lem1_7}
\end{gather}
where we put, with $N\in\mathbb{N}$,
\begin{align}
I_4&=\|F(|x|\geqslant t^{1-2\lambda}|v|/(2(1-2\lambda)))e^{-{\rm i}t^{1-2\lambda}p^2/(2(1-2\lambda))}f(p)\nonumber\\
&\qquad\times F(|x|\leqslant t^{1-2\lambda}|v|/(4(1-2\lambda)))\|,\nonumber\\
I_5&=\|F(|x|>t^{1-2\lambda}|v|/(4(1-2\lambda)))\langle x\rangle^{-N}\|,\nonumber\\
I_6&=\|V^{\rm bdd}(t^\lambda x+t^{1-\lambda}v/(1-2\lambda))F(|x|<t^{1-2\lambda}|v|/(2(1-2\lambda)))\|\label{lem1_8}
\end{align}
as in \eqref{lem1_3}. Using Proposition \ref{enss} for $I_4$, we have
\begin{gather}
\int_{t\geqslant r_0}(I_4+I_5){\rm d}t\lesssim\int_{r_0}^\infty\langle t^{1-2\lambda}v\rangle^{-N}{\rm d}t\nonumber\\
=(|v|^{-1/(1-2\lambda)}/(1-2\lambda))\int_{r_0^{1-2\lambda}|v|}^\infty\langle \tau\rangle^{-N}\tau^{2\lambda/(1-2\lambda)}{\rm d}\tau=O(|v|^{-N}),\label{lem1_9}
\end{gather}
where we changed $\tau=t^{1-2\lambda}|v|$ and chose $N\gg1$ such that $-N+2\lambda/(1-2\lambda)<-1$. As for $I_6$, when $|x|<t^{1-2\lambda}|v|/(2(1-2\lambda))$,
\begin{equation}
|t^\lambda x+t^{1-\lambda}v/(1-2\lambda)|>t^{1-\lambda}|v|/(2(1-2\lambda))\label{lem1_10}
\end{equation}
and
\begin{equation}
I_6\leqslant\|V^{\rm bdd}(x)F(|x|>t^{1-\lambda}|v|/(2(1-2\lambda)))\|
\end{equation}
hold. By the assumption of $V^{\rm bdd}$ \eqref{bdd}, we have
\begin{gather}
\int_{r_0}^\infty I_6{\rm d}t\lesssim\int_{r_0}^\infty\langle t^{1-\lambda}v\rangle^{-\rho}{\rm d}t\\
=(|v|^{-1/(1-\lambda)}/(1-\lambda))\int_{r_0^{1-\lambda}|v|}^\infty\langle \tau\rangle^{-\rho}\tau^{\lambda/(1-\lambda)}{\rm d}\tau=O(|v|^{-\rho}),\label{lem1_11}
\end{gather}
where we changed $\tau=t^{1-\lambda}|v|$ and used $-\rho+\lambda/(1-\lambda)<-1$. Equations \eqref{lem1_4}, \eqref{lem1_5}, \eqref{lem1_6}, \eqref{lem1_9}, and \eqref{lem1_11} imply \eqref{lem1_0}.
\end{proof}

\begin{Lem}\label{lem2}
Let $\Phi_v$ be as in Theorem \ref{the1}. Then
\begin{equation}
\sup_{t\in\mathbb{R}}\|(U(t,0)\tilde{W}^--\tilde{U}_0(t))\Phi_v\|=O(|v|^{-1})
\end{equation}
holds as $|v|\rightarrow\infty$.
\end{Lem}

\begin{proof}
This proof is taken from \cite[Corollary 2.3]{EnWe} (see also \cite{AdKaKaTo}, \cite{AdFuIs}, \cite{AdMa}, \cite{Is1}, \cite{Is2}, \cite{Is3}, \cite{Ni1}, \cite{Ni2}, \cite{Ni3}, \cite{VaWe}, and \cite{We}). We calculate
\begin{gather}
\tilde{W}^--U(t,0)^*\tilde{U}_0(t)=-\int_{-\infty}^t({\rm d}/{\rm d}\tau)U(\tau,0)^*\tilde{U}_0(\tau){\rm d}\tau\nonumber\\
=-{\rm i}\int_{-\infty}^tU(\tau,0)^*V^{\rm bdd}(x)\tilde{U}_0(\tau){\rm d}\tau.
\end{gather}
We thus have
\begin{equation}
\|(\tilde{W}^--U(t,0)^*\tilde{U}_0(t))\Phi_v\|\leqslant\int_{-\infty}^\infty\|V^{\rm bdd}(x)\tilde{U}_0(\tau)\Phi_v\|{\rm d}\tau=O(|v|^{-1})
\end{equation}
as $|v|\rightarrow\infty$ by Lemma \ref{lem1}. This completes the proof.
\end{proof}

\begin{proof}[Proof of Theorem \ref{the1}]
It follows from
\begin{equation}
{\rm i}(\tilde{S}-1)={\rm i}(\tilde{W}^+-\tilde{W}^-)^*\tilde{W}^-=\int_{-\infty}^\infty\tilde{U}_0(t)^*V^{\rm bdd}(x)U(t,0)\tilde{W}^-{\rm d}\tau
\end{equation}
that
\begin{equation}
|v|({\rm i}(\tilde{S}-1)\Phi_v,\Psi_v)=|v|\int_{-\infty}^\infty(V^{\rm bdd}(x)\tilde{U}_0(t)\Phi_v,\tilde{U}_0(t)\Psi_v){\rm d}t+R(v)\label{the1_1}
\end{equation}
where
\begin{equation}
R(v)=|v|\int_{-\infty}^\infty((U(t,0)\tilde{W}^--\tilde{U}_0(t))\Phi_v,V^{\rm bdd}(x)\tilde{U}_0(t)\Psi_v){\rm d}t=O(|v|^{-1})
\end{equation}
as $|v|\rightarrow\infty$ by virtue of Lemmas \ref{lem1} and \ref{lem2}. We separate the integral on the right-hand side of \eqref{the1_1} such that
\begin{equation}
\int_{-\infty}^\infty=\int_{|t|<\pi/(2\omega)}+\int_{\pi/(2\omega)\leqslant|t|<r_0}+\int_{|t|\geqslant r_0}\label{the1_2}
\end{equation}
and consider the part $|t|<\pi/(2\omega)$ first. By \eqref{mehler2} and \eqref{lem1_1}, we have
\begin{equation}
e^{-{\rm i}v\cdot x}e^{{\rm i}tH_0}V^{\rm bdd}(x)e^{-{\rm i}tH_0}e^{{\rm i}v\cdot x}=e^{{\rm i}tH_0}V^{\rm bdd}(x+\sin\omega tv/\omega)e^{-{\rm i}tH_0}.\label{the1_3}
\end{equation}
We thus have
\begin{gather}
|v|\int_{|t|<\pi/(2\omega)}(V^{\rm bdd}(x)e^{-{\rm i}tH_0}\Phi_v,e^{-{\rm i}tH_0}\Psi_v){\rm d}t\nonumber\\
=|v|\int_{|t|<\pi/(2\omega)}(V^{\rm bdd}(x+\sin\omega tv/\omega)e^{-{\rm i}tH_0}\Phi_0,e^{-{\rm i}tH_0}\Psi_0){\rm d}t\nonumber\\
=\int_{|\tau|<|v|/\omega}(1/\sqrt{1-(\omega\tau/|v|)^2})(V^{\rm bdd}(x+\hat{v}\tau)e^{-{\rm i}\arcsin(\omega\tau/|v|)H_0/\omega}\Phi_0,\nonumber\\
e^{-{\rm i}\arcsin(\omega\tau/|v|)H_0/\omega}\Psi_0){\rm d}\tau\label{the1_4}
\end{gather}
by changing $\tau=\sin\omega t|v|/\omega$. Because $e^{-{\rm i}tH_0}$ is strongly continuous at $t=0$, we have
\begin{gather}
(1/\sqrt{1-(\omega\tau/|v|)^2})(V^{\rm bdd}(x+\hat{v}\tau)e^{-{\rm i}\arcsin(\omega\tau/|v|)H_0/\omega}\Phi_0,e^{-{\rm i}\arcsin(\omega\tau/|v|)H_0/\omega}\Psi_0)\nonumber\\
\rightarrow(V^{\rm bdd}(x+\hat{v}\tau)\Phi_0,\Psi_0)
\end{gather}
as $|v|\rightarrow\infty$ pointwisely in $\tau\in\mathbb{R}$. In addition, we have
\begin{gather}
|v|\int_{|t|<\pi/(2\omega)}|(V^{\rm bdd}(x)e^{-{\rm i}tH_0}\Phi_v,e^{-{\rm i}tH_0}\Psi_v)|{\rm d}t\nonumber\\
=\int_{|\tau|<\pi|v|/(2\omega)}|(V^{\rm bdd}(x)e^{-{\rm i}(\tau/|v|)H_0}\Phi_v,e^{-{\rm i}(\tau/|v|)H_0}\Psi_v)|{\rm d}\tau
\end{gather}
by changing $\tau=|v|t$. It follows from the calculations in the proof of Lemma \ref{lem1} that
\begin{gather}
|(V^{\rm bdd}(x)e^{-{\rm i}(\tau/|v|)H_0}\Phi_v,e^{-{\rm i}(\tau/|v|)H_0}\Psi_v)|\leqslant\|V^{\rm bdd}(x)e^{-{\rm i}(\tau/|v|)H_0}\Phi_v\|\|\Psi_0\|\nonumber\\
\lesssim\langle\tan\omega(\tau/|v|)v\rangle^{-2}+\langle\sin\omega(\tau/|v|)v\rangle^{-\rho}\lesssim\langle\tau\rangle^{-2}+\langle\tau\rangle^{-\rho}.
\end{gather}
We therefore obtain
\begin{equation}
|v|\int_{|t|<\pi/(2\omega)}(V^{\rm bdd}(x)e^{-{\rm i}tH_0}\Phi_v,e^{-{\rm i}tH_0}\Psi_v){\rm d}t\rightarrow\int_{-\infty}^\infty(V^{\rm bdd}(x+\hat{v}\tau)\Phi_0,\Psi_0){\rm d}\tau\label{the1_5}
\end{equation}
as $|v|\rightarrow\infty$ by the Lebesgue dominated convergence theorem. To complete our proof, we prove that the second and third integrals of \eqref{the1_2} converge to zero as $|v|\rightarrow\infty$. These were almost proved already in Lemma \ref{lem1}. Indeed, for the second integral over $\pi/(2\omega)\leqslant|t|<r_0$, we find that $|(V^{\rm bdd}(x)e^{-{\rm i}tH_0}\Phi_v,e^{-{\rm i}tH_0}\Psi_v)|\leqslant\|V^{\rm bdd}(x)e^{-{\rm i}tH_0}\Phi_v\|\|\Psi_0\|$ and
\begin{gather}
|v|\int_{\pi/(2\omega)\leqslant|t|<r_0}|(V^{\rm bdd}(x)e^{-{\rm i}tH_0}\Phi_v,e^{-{\rm i}tH_0}\Psi_v)|{\rm d}t\nonumber\\
\lesssim|v|\int_{\pi/(2\omega)}^{r_0}(\langle\tan\omega tv\rangle^{-2}+\langle\sin\omega tv\rangle^{-\rho}){\rm d}t=O(|v|^{-1})+O(|v|^{-\rho+1})\label{the1_6}
\end{gather}
by \eqref{lem1_5} and \eqref{lem1_6}. For the third integral on $|t|\geqslant r_0$, we find that
\begin{equation}
|v|\int_{|t|\geqslant r_0}|(V^{\rm bdd}(x)\tilde{U}_0(t)\Phi_v,\tilde{U}_0(t)\Psi_v)|{\rm d}t=O(|v|^{-N+1})+O(|v|^{-\rho+1})\label{the1_7}
\end{equation}
by \eqref{lem1_9} and \eqref{lem1_11}. With $N\geqslant2$ and $\rho>1$, equations \eqref{the1_5}, \eqref{the1_6} and \eqref{the1_7} complete the proof.
\end{proof}

\section{Singular case}
We now consider the instances $V^{\rm sing}\not\equiv0$ and prove the following reconstruction formula. At the end of this section, we finally complete the proof of Theorem \ref{main_thm}.

\begin{Thm}\label{the2}
Let $\Phi_v$ and $\Psi_v$ be as in Theorem \ref{the1}. Then
\begin{equation}
\lim_{|v|\rightarrow\infty}|v|({\rm i}(\tilde{S}(V)-1)\Phi_v,\Psi_v)=\int_{-\infty}^\infty(V(x+\hat{v}t)\Phi_0,\Psi_0){\rm d}t
\end{equation}
holds.
\end{Thm}

To prove Theorem \ref{the2}, we prepare the following Lemma \ref{lem3}, which is the singular version of Lemma \ref{lem1}.
\begin{Lem}\label{lem3}
Let $\Phi_v$ be as in Theorem \ref{the1}. Then
\begin{equation}
\int_{-\infty}^{\infty}\|V^{\rm sing}(x)\tilde{U}_0(t)\Phi_v\|{\rm d}t=O(|v|^{-1})\label{lem3_0}
\end{equation}
holds as $|v|\rightarrow\infty$.
\end{Lem}

\begin{proof}
As in the proof of Lemma \ref{lem1}, we take $f\in C_0^\infty(\mathbb{R}^n)$ such that $\Phi_0=f(p)\Phi_0$ and $\supp f\subset\{\xi\in\mathbb{R}^n\bigm||\xi|\leqslant\eta\}$ with some $\eta>0$ and assume that \eqref{r_0}. Separating the integral such that
\begin{equation}
\int_{-\infty}^{\infty}=\int_{|t|\leqslant\pi/(4\omega)}+\int_{\pi/(4\omega)<|t|<r_0}+\int_{|t|\geqslant r_0}
\end{equation}
and first consider the part $|t|\leqslant\pi/(4\omega)$. Similar to \eqref{lem1_2} and \eqref{lem1_3}, we have
\begin{gather}
\|V^{\rm sing}(x)e^{-{\rm i}tH_0}\Phi_v\|=\|V^{\rm sing}(\cos\omega tx+\sin\omega tv/\omega)\langle p/\cos\omega t\rangle^{-2}\nonumber\\
\times e^{-{\rm i}\tan\omega tp^2/(2\omega)}\langle p/\cos\omega t\rangle^2\Phi_0\|\leqslant I_1+I_2+I_3,
\end{gather}
where we put
\begin{align}
I_1&=\|V^{\rm sing}(\cos\omega tx+\sin\omega tv/\omega)\langle p/\cos\omega t\rangle^{-2}\|\|\langle p/\cos\omega t\rangle^2\Phi_0\|\nonumber\\
&\quad\times\|F(|x|\geqslant|\tan\omega t||v|/(2\omega))e^{-{\rm i}\tan\omega tp^2/(2\omega)}f(p)F(|x|\leqslant|\tan\omega t||v|/(4\omega))\|,\nonumber\\
I_2&=\|V^{\rm sing}(\cos\omega tx+\sin\omega tv/\omega)\langle p/\cos\omega t\rangle^{-2}\|\|F(|x|\geqslant|\tan\omega t||v|/(2\omega))\nonumber\\
&\quad\times e^{-{\rm i}\tan\omega tp^2/(2\omega)}f(p)F(|x|>|\tan\omega t||v|/(4\omega))\langle x\rangle^{-2}\|\|\langle x\rangle^2\langle p/\cos\omega t\rangle^2\Phi_0\|,\nonumber\\
I_3&=\|V^{\rm sing}(\cos\omega tx+\sin\omega tv/\omega)\langle p/\cos\omega t\rangle^{-2}\nonumber\\
&\quad\times F(|x|<|\tan\omega t||v|/(2\omega))\|\|\langle p/\cos\omega t\rangle^2\Phi_0\|\label{lem3_1}
\end{align}
as in the proof of \cite[Proposition 2.3]{Is1}. Noting that
\begin{gather}
\|V^{\rm sing}(\cos\omega tx+\sin\omega tv/\omega)\langle p/\cos\omega t\rangle^{-2}\|\nonumber\\
=\|V^{\rm sing}(\cos\omega tx)\langle p/\cos\omega t\rangle^{-2}\|=\|V^{\rm sing}(x)\langle p\rangle^{-2}\|
\end{gather}
and that 
\begin{equation}
\|\langle p/\cos\omega t\rangle^2\Phi_0\|\leqslant\|\langle\sqrt{2} p\rangle^2\Phi_0\|
\end{equation}
because $|t|\leqslant\pi/(4\omega)$, we have
\begin{equation}
\int_{|t|\leqslant\pi/(4\omega)}(I_1+I_2){\rm d}t\lesssim\int_0^{\pi/(4\omega)}\langle\omega tv\rangle^{-2}{\rm d}t=O(|v|^{-1})\label{lem3_2}
\end{equation}
as in the proof of Lemma \ref{lem1}. Because
\begin{equation}
|\cos\omega tx+\sin\omega tv/\omega|>|\sin\omega t||v|/(2\omega)\geqslant|t||v|/4
\end{equation}
holds when $|x|<|\tan\omega t||v|/(2\omega)$, we have
\begin{gather}
\int_{|t|\leqslant\pi/(4\omega)}I_3{\rm d}t\lesssim\int_0^{\pi/(4\omega)}\|V^{\rm sing}(x)\langle p\rangle^{-2}F(|x|\geqslant|v|t/4)\|{\rm d}t\nonumber\\
=|v|^{-1}\int_0^1+|v|^{-1}\int_1^{\pi|v|/(4\omega)}\|V^{\rm sing}(x)\langle p\rangle^{-2}F(|x|\geqslant\tau/4)\|{\rm d}\tau
\end{gather}
by changing $\tau=|v|t$. The first integral over interval $0\leqslant\tau<1$ clearly has order $O(|v|^{-1})$. For the second integral over $1\leqslant\tau\leqslant\pi|v|/(4\omega)$, we take $\chi\in C^\infty(\mathbb{R}^n)$ such that $\chi(x)=1$ if $|x|\geqslant1$ and $\chi(x)=0$ if $|x|\leqslant1/2$. We then have
\begin{gather}
\|V^{\rm sing}(x)\langle p\rangle^{-2}F(|x|\geqslant\tau/4)\|\leqslant\|V^{\rm sing}(x)\langle p\rangle^{-2}\chi(4x/\tau)\|\nonumber\\
\lesssim\|V^{\rm sing}(x)\chi(4x/\tau)\langle p\rangle^{-2}\|+\tau^{-1}\|V^{\rm sing}(x)(\nabla\chi)(4x/\tau)\langle p\rangle^{-2}\|+\tau^{-2}\|V^{\rm sing}(x)\langle p\rangle^{-2}\|\label{lem3_3}
\end{gather}
by calculating the commutator
\begin{equation}
{\rm i}[\langle p\rangle^{-2},\chi(4x/\tau)]=8\langle p\rangle^{-2}((\nabla\chi)(4x/\tau)\cdot p/\tau-2{\rm i}(\Delta\chi)(4x/\tau)/\tau^2)\langle p\rangle^{-2}.
\end{equation}
Noting that $V^{\rm sing}$ is compactly supported and that the integral intervals of the first and second terms of \eqref{lem3_3} are finite for $|v|\gg1$, we have
\begin{equation}
\int_{|t|\leqslant\pi/(4\omega)}I_3{\rm d}t=O(|v|^{-1}).\label{lem3_4}
\end{equation}
We next consider the integral over $\pi/(4\omega)<|t|<r_0$. The strategy for the estimates of the integral stems from the proof of \cite[Proposition 2.3]{Is1} (see also \cite[Lemma 3.1]{Is4} and \cite[Lemma 4]{Ni3}). Using the same computation with \eqref{the1_3} and the Mehler formula \eqref{mehler1}, we have
\begin{gather}
\|V^{\rm sing}(x)e^{-{\rm i}tH_0}\Phi_v\|=\|V^{\rm sing}(x+\sin\omega tv/\omega)e^{-{\rm i}tH_0}\Phi_0\|\nonumber\\
=\|V^{\rm sing}(\sin\omega t(x+v)/\omega)\mathscr{F}\mathscr{M}(\tan\omega t/\omega)\Phi_0\|\leqslant I_4+I_5,
\end{gather}
where we put
\begin{align}
I_4&=\|V^{\rm sing}(\sin\omega t(x+v)/\omega)\langle\omega p/\sin\omega t\rangle^{-2}F(|x|\leqslant|v|/2)\|\|\langle\omega x/\sin\omega t\rangle^{2}\Phi_0\|,\\
I_5&=\|V^{\rm sing}(\sin\omega t(x+v)/\omega)\langle\omega p/\sin\omega t\rangle^{-2}\|\nonumber\\
&\quad\times\|F(|x|>|v|/2)\mathscr{F}\mathscr{M}(\tan\omega t/\omega)\langle\omega x/\sin\omega t\rangle^{2}\Phi_0\|.\quad
\end{align}
Clearly
\begin{equation}
\|\langle\omega x/\sin\omega t\rangle^{2}\Phi_0\|\lesssim\|\langle x\rangle^2\Phi_0\|
\end{equation}
holds because $\pi/(4\omega)<|t|<r_0$ and 
\begin{equation}
0<\min\{1/\sqrt{2},\sin\omega r_0\}<|\sin\omega t|
\end{equation}
noting $r_0<\pi/\omega$. When $|x|\leqslant|v|/2$, there exists a small constant $c>0$ such that
\begin{equation}
|\sin\omega t(x+v)/\omega|\geqslant|\sin\omega t||v|/(2\omega)\geqslant c|t||v|
\end{equation}
again noting $r_0<\pi/\omega$. We thus have
\begin{gather}
\int_{\pi/(4\omega)<|t|<r_0}I_4{\rm d}t\lesssim\int_{\pi/(4\omega)}^{r_0}\|V^{\rm sing}(x)\langle p\rangle^{-2}F(|x|\geqslant ct|v|)\|{\rm d}t\nonumber\\
=|v|^{-1}\int_{\pi|v|/(4\omega)}^{r_0|v|}\|V^{\rm sing}(x)\langle p\rangle^{-2}F(|x|\geqslant c\tau)\|{\rm d}\tau
\end{gather}
by changing $\tau=|v|t$. For $\tau>\pi|v|/(4\omega)\gg1$, we have
\begin{equation}
\|V^{\rm sing}(x)\langle p\rangle^{-2}F(|x|\geqslant c\tau)\|\lesssim\tau^{-2}\|V^{\rm sing}(x)\langle p\rangle^{-2}\|\label{lem3_5}
\end{equation}
as in \eqref{lem3_3} noting that $V^{\rm sing}$ is compactly supported. Therefore, we can obtain
\begin{equation}
\int_{\pi/(4\omega)<|t|<r_0}I_4{\rm d}t=O(|v|^{-2}).\label{lem3_6}
\end{equation}
For the integral $I_5$, we write
\begin{gather}
\mathscr{F}\mathscr{M}(\tan\omega t/\omega)\langle\omega x/\sin\omega t\rangle^{2}\Phi_0\nonumber\\
=\int_{\mathbb{R}^n}e^{-{\rm i}x\cdot y}e^{{\rm i}\omega y^2/(2\tan\omega t)}\langle\omega y/\sin\omega t\rangle^{2}\Phi_0(y){\rm d}y/(2\pi)^{n/2}.
\end{gather}
Using the relation $e^{-{\rm i}x\cdot y}=\langle x\rangle^{-2}(1+{\rm i}x\cdot \nabla_y)e^{-{\rm i}x\cdot y}$ and integrating by parts,
we have
\begin{gather}
\mathscr{F}\mathscr{M}(\tan\omega t/\omega)\langle\omega x/\sin\omega t\rangle^{2}\Phi_0=\langle x\rangle^{-2}\mathscr{F}\mathscr{M}(\tan\omega t/\omega)\langle\omega x/\sin\omega t\rangle^{2}\Phi_0\nonumber\\
+(\omega/\tan\omega t)\langle x\rangle^{-2}x\cdot\mathscr{F}x\mathscr{M}(\tan\omega t/\omega)\langle\omega x/\sin\omega t\rangle^{2}\Phi_0\nonumber\\
-{\rm i}\langle x\rangle^{-2}x\cdot\mathscr{F}\mathscr{M}(\tan\omega t/\omega)\nabla_x\langle\omega x/\sin\omega t\rangle^{2}\Phi_0\label{lem3_7}
\end{gather}
and
\begin{gather}
\|F(|x|>|v|/2)\mathscr{F}\mathscr{M}(\tan\omega t/\omega)\langle\omega x/\sin\omega t\rangle^{2}\Phi_0\|\nonumber\\
\lesssim |v|^{-2}\|\langle x\rangle^2\Phi_0\|+|v|^{-1}(\|\langle x\rangle^3\Phi_0\|+\|\langle x\rangle^2\nabla\Phi_0\|).\label{lem3_8}
\end{gather}
It follows from \eqref{lem3_8} and
\begin{equation}
\|V^{\rm sing}(\sin\omega t(x+v)/\omega)\langle\omega p/\sin\omega t\rangle^{-2}\|=\|V^{\rm sing}(x)\langle p\rangle^{-2}\|
\end{equation}
that
\begin{equation}
\int_{\pi/(4\omega)<|t|<r_0}I_5{\rm d}t=O(|v|^{-1}).\label{lem3_9}
\end{equation}
We consider the final integral over $|t|\geqslant r_0$, in particular $t\geqslant r_0$. In the same way with \eqref{lem1_7}, \eqref{lem1_8}, and \eqref{lem3_1}, we have
\begin{gather}
\|V^{\rm sing}(x)\tilde{U}_0(t)\Phi_v\|=\|V^{\rm sing}(t^\lambda x+t^{1-\lambda}v/(1-2\lambda))\langle p/t^\lambda\rangle^{-2}\nonumber\\
\times e^{-{\rm i}t^{1-2\lambda}p^2/(2(1-2\lambda))}\langle p/t^\lambda\rangle^2\Phi_0\|\lesssim I_6+I_7+I_8,
\end{gather}
where we put, with $N\in\mathbb{N}$,
\begin{align}
I_6&=\|F(|x|\geqslant t^{1-2\lambda}|v|/(2(1-2\lambda)))e^{-{\rm i}t^{1-2\lambda}p^2/(2(1-2\lambda))}f(p)\nonumber\\
&\qquad\times F(|x|\leqslant t^{1-2\lambda}|v|/(4(1-2\lambda)))\|,\nonumber\\
I_7&=\|F(|x|>t^{1-2\lambda}|v|/(4(1-2\lambda)))\langle x\rangle^{-N}\|,\nonumber\\
I_8&=\|V^{\rm sing}(t^\lambda x+t^{1-\lambda}v/(1-2\lambda))\langle p/t^\lambda\rangle^{-2}F(|x|<t^{1-2\lambda}|v|/(2(1-2\lambda)))\|.\label{lem3_10}
\end{align}
We here used $\|V^{\rm sing}(t^\lambda x)\langle p/t^\lambda\rangle^{-2}\|=\|V^{\rm sing}(x)\langle p\rangle^{-2}\|$, $\|\langle p/t^\lambda\rangle^2\Phi_0\|\leqslant\|\langle p/r_0^\lambda\rangle^2\Phi_0\|$ and
\begin{equation}
\|\langle x\rangle^N\langle p/t^\lambda\rangle^2\Phi_0\|\leqslant\|\langle p/t^\lambda\rangle^2\langle x\rangle^N\Phi_0\|+\|[\langle x\rangle^N,\langle p/t^\lambda\rangle^2]\Phi_0\|\lesssim1
\end{equation}
in \eqref{lem3_10}. We immediately have
\begin{equation}
\int_{r_0}^\infty(I_6+I_7){\rm d}t=O(|v|^{-N})\label{lem3_11}
\end{equation}
as in \eqref{lem1_9} for $N\gg1$ such that $-N+2\lambda/(1-2\lambda)<-1$. Because \eqref{lem1_10} holds when $|x|<t^{1-2\lambda}|v|/(2(1-2\lambda))$, we have
\begin{gather}
\int_{r_0}^\infty I_8{\rm d}t\leqslant\int_{r_0}^\infty\|V^{\rm sing}(x)\langle p\rangle^{-2}F(|x|>t^{1-\lambda}|v|/(2(1-2\lambda)))\|{\rm d}t\nonumber\\
\lesssim|v|^{-1/(1-\lambda)}\int_{r_0^{1-\lambda}|v|}^\infty\tau^{\lambda/(1-\lambda)}\|V^{\rm sing}(x)\langle p\rangle^{-2}F(|x|>\tau/(2(1-2\lambda)))\|{\rm d}\tau
\end{gather}
by changing $\tau=t^{1-\lambda}|v|$. As in \eqref{lem3_5}, we thus have
\begin{equation}
\int_{r_0}^\infty I_8{\rm d}t\lesssim|v|^{-1/(1-\lambda)}\int_{r_0^{1-\lambda}|v|}^\infty\tau^{\lambda/(1-\lambda)-2}{\rm d}\tau=O(|v|^{-2})\label{lem3_12}
\end{equation}
noting that $\lambda/(1-\lambda)-2<-1$. Equations \eqref{lem3_2}, \eqref{lem3_4}, \eqref{lem3_6}, \eqref{lem3_9}, \eqref{lem3_11}, and \eqref{lem3_12} imply \eqref{lem3_0}.
\end{proof}

\begin{proof}[Proof of Theorem \ref{the2}]
Note that Lemma \ref{lem2} also holds for $V=V^{\rm bdd}+V^{\rm sing}$ by virtue of Lemma \ref{lem3}. We therefore have
\begin{equation}
|v|({\rm i}(\tilde{S}-1)\Phi_v,\Psi_v)=|v|\int_{-\infty}^\infty(V(x)\tilde{U}_0(t)\Phi_v,\tilde{U}_0(t)\Psi_v){\rm d}t+O(|v|^{-1}).
\end{equation}
Because we have already proved
\begin{equation}
|v|\int_{-\infty}^\infty(V^{\rm bdd}(x)\tilde{U}_0(t)\Phi_v,\tilde{U}_0(t)\Psi_v){\rm d}t\rightarrow\int_{-\infty}^\infty(V^{\rm bdd}(x+\hat{v}\tau)\Phi_0,\Psi_0){\rm d}\tau
\end{equation}
as $|v|\rightarrow\infty$ in the proof of Theorem \ref{the1}, it suffices to prove
\begin{equation}
|v|\int_{-\infty}^\infty(V^{\rm sing}(x)\tilde{U}_0(t)\Phi_v,\tilde{U}_0(t)\Psi_v){\rm d}t\rightarrow\int_{-\infty}^\infty(V^{\rm sing}(x+\hat{v}\tau)\Phi_0,\Psi_0){\rm d}\tau\label{the2_1}
\end{equation}
as $|v|\rightarrow\infty$. We separate the integral such that
\begin{equation}
\int_{-\infty}^{\infty}=\int_{|t|\leqslant\pi/(4\omega)}+\int_{\pi/(4\omega)<|t|<r_0}+\int_{|t|\geqslant r_0}
\end{equation}
and first consider the integral over $|t|\leqslant\pi/(4\omega)$. As in \eqref{the1_4}, we have
\begin{gather}
|v|\int_{|t|\leqslant\pi/(4\omega)}(V^{\rm sing}(x)e^{-{\rm i}tH_0}\Phi_v,e^{-{\rm i}tH_0}\Psi_v){\rm d}t=\int_{|\tau|\leqslant|v|/(\sqrt{2}\omega)}(1/\sqrt{1-(\omega\tau/|v|)^2})\nonumber\\
\times(V^{\rm sing}(x+\hat{v}\tau)e^{-{\rm i}\arcsin(\omega\tau/|v|)H_0/\omega}\Phi_0,e^{-{\rm i}\arcsin(\omega\tau/|v|)H_0/\omega}\Psi_0){\rm d}\tau.
\end{gather}
Because
\begin{equation}
p_je^{-{\rm i}tH_0}\Phi_0=-\sin\omega te^{-{\rm i}tH_0}x_j\Phi_0-\sin\omega t\tan\omega te^{-{\rm i}tH_0}p_j\Phi_0/\omega+\omega e^{-{\rm i}tH_0}p_j\Phi_0/\cos\omega t
\end{equation}
for $1\leqslant j\leqslant n$ by \eqref{mehler2}, we have
\begin{equation}
\|\langle p\rangle^2e^{-{\rm i}\arcsin(\omega\tau/|v|)H_0/\omega}\Phi_0\|\lesssim1
\end{equation}
for $|\tau|\leqslant|v|/(\sqrt{2}\omega)$ and
\begin{gather}
(1/\sqrt{1-(\omega\tau/|v|)^2})(V^{\rm sing}(x+\hat{v}\omega\tau)e^{-{\rm i}\arcsin(\omega\tau/|v|)H_0/\omega}\Phi_0,e^{-{\rm i}\arcsin(\omega\tau/|v|)H_0/\omega}\Psi_0)\nonumber\\
\rightarrow(V^{\rm sing}(x+\hat{v}\tau)\Phi_0,\Psi_0)
\end{gather}
as $|v|\rightarrow\infty$ pointwisely in $\tau\in\mathbb{R}$. In addition, we have
\begin{gather}
|v|\int_{|t|\leqslant\pi/(4\omega)}|(V^{\rm sing}(x)e^{-{\rm i}tH_0}\Phi_v,e^{-{\rm i}tH_0}\Psi_v)|{\rm d}t\nonumber\\
=\int_{|\tau|\leqslant\pi|v|/(4\omega)}|(V^{\rm sing}(x)e^{-{\rm i}(\tau/|v|)H_0}\Phi_v,e^{-{\rm i}(\tau/|v|)H_0}\Psi_v)|{\rm d}\tau
\end{gather}
by changing $\tau=|v|t$. From the calculations developed in the proof of Lemma \ref{lem3}, we find
\begin{gather}
|(V^{\rm sing}(x)e^{-{\rm i}(\tau/|v|)H_0}\Phi_v,e^{-{\rm i}(\tau/|v|)H_0}\Psi_v)|\leqslant\|V^{\rm sing}(x)e^{-{\rm i}(\tau/|v|)H_0}\Phi_v\|\|\Psi_0\|\nonumber\\
\lesssim\langle\tau\rangle^{-2}+\|V^{\rm sing}(x)\langle p\rangle^{-2}F(|x|\geqslant|\tau|/4)\|.\label{the2_2}
\end{gather}
The right-hand side of \eqref{the2_2} is integrable for $\tau$ independently of $v$ (see \eqref{lem3_3}). We therefore obtain
\begin{equation}
|v|\int_{|t|\leqslant\pi/(4\omega)}(V^{\rm sing}(x)e^{-{\rm i}tH_0}\Phi_v,e^{-{\rm i}tH_0}\Psi_v){\rm d}t\rightarrow\int_{-\infty}^\infty(V^{\rm sing}(x+\hat{v}\tau)\Phi_0,\Psi_0){\rm d}\tau\label{the2_3}
\end{equation}
as $|v|\rightarrow\infty$ by the Lebesgue dominated convergence theorem. For the integral over $\pi/(4\omega)<|t|<r_0$, integrating by parts in \eqref{lem3_7} once more, we find that \eqref{lem3_9} has order $O(|v|^{-2})$. we thus have
\begin{equation}
|v|\int_{\pi/(4\omega)<|t|<r_0}|(V^{\rm sing}(x)e^{-{\rm i}tH_0}\Phi_v,e^{-{\rm i}tH_0}\Psi_v)|{\rm d}t=O(|v|^{-1})\label{the2_4}
\end{equation}
as $|v|\rightarrow\infty$ by using calculations obtained in the proof of Lemma \ref{lem3} (see also \eqref{lem3_6}). Finally, for the integral over $|t|\geqslant r_0$, we also have
\begin{equation}
|v|\int_{|t|\geqslant r_0}|(V^{\rm sing}(x)\tilde{U}_0(t)\Phi_v,\tilde{U}_0(t)\Psi_v)|{\rm d}t=O(|v|^{-N+1})+O(|v|^{-1})\label{the2_5}
\end{equation}
as $|v|\rightarrow\infty$ by \eqref{lem3_11} and \eqref{lem3_12}. With $N\geqslant2$, equations \eqref{the2_3}, \eqref{the2_4} and \eqref{the2_5} imply \eqref{the2_1}.
\end{proof}

\begin{Rem}\label{rem2}
In our proofs of Theorems \ref{the1}, \ref{the2}, Lemmas \ref{lem1} and \ref{lem3}, we partitioned the integrals at points $\pi/(4\omega)$, $\pi/(2\omega)$, and $r_0$. However, if we assume $0<r_0<\pi/(2\omega)$, it suffices to separate the integrals such that
\begin{equation}
\int_{-\infty}^\infty=\int_{|t|<r_0}+\int_{|t|\geqslant r_0}
\end{equation}
in these proofs. We especially do not have to consider the integrals over $\pi/(4\omega)<|t|<r_0$ in the proofs of Theorem \ref{the2} and Lemma \ref{lem3} even if $r_0>\pi/(4\omega)$.
\end{Rem}

\begin{proof}[Proof of Theorem \ref{main_thm}]
By the same computation with \eqref{lem3_3}, we have the Enss condition
\begin{equation}
\int_0^\infty\|V(x)\langle p\rangle^{-2}F(|x|\geqslant R)\|dR<\infty.
\end{equation}
We can assume $\tilde{S}(V_1)=\tilde{S}(V_2)$ by \eqref{tilde_S}. From Theorem \ref{the2} and the Plancherel formula associated with the Radon transform (\cite[Theorem 2.17 in Chap.1]{He}), $V_1=V_2$ can be proved similarly as in the proof of \cite[Theorem 1.1 in p. 3909]{EnWe}.
\end{proof}

\noindent\textbf{Acknowledgments.} 
This work was supported by JSPS KAKENHI Grant Numbers JP20K03625 and JP21K03279. The author thanks Associate Professor Masaki Kawamoto of Okayama University and referees for their valuable comments.

\end{document}